\documentclass[sigconf,nonacm]{AAMAS/aamas}

\usepackage{balance} 
\usepackage{mathtools, amssymb} 
\usepackage{algorithm}
\usepackage[noend]{algpseudocode}
\usepackage{thmtools,thm-restate}
\usepackage{graphicx} 
\usepackage{svg}
\usepackage[colorinlistoftodos]{todonotes}
\usepackage{subcaption}
\usepackage{adjustbox}
\usepackage{hyperref}
\usepackage{centernot}
\usepackage{cleveref}

\newtheorem{theorem}{Theorem}
\newtheorem{remark}[theorem]{Remark}

\newtheorem{lemma}[theorem]{Lemma}
\newtheorem{corollary}[theorem]{Corollary}
\newtheorem{example}[theorem]{Example}

\newtheorem{definition}[theorem]{Definition}

\def\thmt@innercounters{equation,algocf}

\newcommand{\Psat}[1][P]{$#1$-\textsc{Satisfiability}}
\newcommand{\PsatT}[1][P]{$#1$-Satisfiability} 
\newcommand{\Pmksp}[1][P]{$#1$-\textsc{Makespan}}
\newcommand{\PmkspT}[1][P]{$#1$-Makespan} 
\newcommand{\supp}{\mathrm{supp}}
\newcommand{\ww}{\mathbf{w}} 
\newcommand{\ba}{\mathbf{a}}
\newcommand{\bs}{\mathbf{s}} 
\newcommand{\outpt}[1]{\textsf{#1}} 

\usepackage[framemethod=TikZ]{mdframed}
\newmdenv[
  skipabove=\baselineskip,
  skipbelow=\baselineskip,
  linewidth=0.6pt,
  roundcorner=6pt,
  linecolor=black!40,
  backgroundcolor=black!2,
  innerleftmargin=8pt,
  innerrightmargin=8pt,
  innertopmargin=6pt,
  innerbottommargin=6pt
]{problemBox}

\newcommand{\problem}[3]{%
  \begin{problemBox}%
  \noindent\textbf{Problem:} \textsc{#1}\\[2pt]%
  \textbf{Input:} #2\\[2pt]%
  \textbf{Task:} #3%
  \end{problemBox}%
}

\setcopyright{ifaamas}
\acmConference[AAMAS '26]{Proc.\@ of the 25th International Conference
on Autonomous Agents and Multiagent Systems (AAMAS 2026)}{May 25 -- 29, 2026}
{Paphos, Cyprus}{C.~Amato, L.~Dennis, V.~Mascardi, J.~Thangarajah (eds.)}
\copyrightyear{2026}
\acmYear{2026}
\acmDOI{}
\acmPrice{}
\acmISBN{}

\acmSubmissionID{238}

\title{Fair Coordination in Strategic Scheduling}

\author{Wei-Chen Lee}
\affiliation{
  \institution{University of Oxford}
  \city{Oxford}
  \country{United Kingdom}}
\email{wei-chen.lee@cs.ox.ac.uk}

\author{Martin Bullinger}
\affiliation{
  \institution{University of Bristol}
  \city{Bristol}
  \country{United Kingdom}}
\email{martin.bullinger@bristol.ac.uk}

\author{Alessandro Abate}
\affiliation{
  \institution{University of Oxford}
  \city{Oxford}
  \country{United Kingdom}}
\email{alessandro.abate@cs.ox.ac.uk}

\author{Michael Wooldridge}
\affiliation{
  \institution{University of Oxford}
  \city{Oxford}
  \country{United Kingdom}}
\email{michael.wooldridge@cs.ox.ac.uk}

\begin{abstract}
    We consider a scheduling problem of strategic agents representing jobs of different weights.
    Each agent has to decide on one of a finite set of identical machines to get their job processed. 
    In contrast to the common and exclusive focus on makespan minimization, we want the outcome to be fair under strategic considerations of the agents. 
    Two natural properties are credibility, which ensures that the assignment is a Nash equilibrium and equality, requiring that agents with equal-weight jobs are assigned to machines of equal load.
    We combine these two with a hierarchy of fairness properties based on envy-freeness together with several relaxations based on the idea that envy seems more justified towards agents with a higher weight.
    We present a complete complexity landscape for satisfiability and decision versions of these properties, alone or in combination, and study them as structural constraints under makespan optimization.
    For our positive results, we develop a unified algorithmic approach, where we achieve different properties by fine-tuning key subroutines.
\end{abstract}


\newcommand{\BibTeX}{\rm B\kern-.05em{\sc i\kern-.025em b}\kern-.08em\TeX}

\begin{document}


\pagestyle{fancy}
\fancyhead{}


\maketitle 


\section{Introduction} \label{sec:intro}


We study a strategic variant of a classic problem known as the identical machine scheduling problem \cite{Graham69Schedulling}, where each of the $n$ weighted tasks is assigned to one of the $m$ identical resources (machines). The load on a resource is the sum of the weights of the tasks assigned to it, and the objective of the assignment is to minimize the largest load on any resource, known as the `makespan' of the assignment.\footnote{The interpretation is to consider the weight of a task as the time taken to process it, and thus the makespan is the time required to process all tasks.} 
Strategic variants of this problem, where each task is owned by a different player concerned only with the completion time of their task, or the completion time of the resource to which their task is assigned, are known as the selfish load balancing problem or the weighted singleton congestion game.\footnote{Different definitions of the cost incurred by a player have been studied, with some defined as the time taken until the completion of their task, while others as the time taken to complete all tasks assigned to the same resource. The latter definition is closer to the notion of congestion and is the one we use in this paper.}

In a strategic variant where players can coordinate their choices, fairness is an important consideration among the players.
A natural idea in this context is envy-freeness \cite{Fole67a,Vari74a}, a standard notion in fair division \cite{BrTa96a}, which requires that no agent would rather replace another agent on their assigned resource.
However, naive envy-freeness is often impossible to satisfy,\footnote{An example is the instance discussed in \Cref{fig:assignment_examples}.}
and even feels unnecessarily rigid as agents congest machines differently.
This motivates the relaxation of envy-freeness where envy is only justified towards agents representing a heavier job.
We introduce a hierarchy of such notions and study them along two other natural desiderata: credibility, a Nash equilibrium condition that demands that no player can benefit from a unilateral deviation, and equality, another fairness notion that requires agents of equal weight to be assigned to machines of equal load.
We then consider the algorithmic challenge of computing desirable assignments that satisfy single or several of these properties.

Notably, our properties can be seen as natural constraints for the identical machine scheduling problem. 
Interestingly, (the combination of) several of them render this problem polynomial-time solvable.
We view it as an important message from our paper that natural constraints can facilitate computational feasibility.

\subsection{Related Literature}

Early works in the scheduling literature establish that the base problem of identical machine scheduling (without further constraints)
is NP-complete \cite{Garey1979IMS_NP-complete}, and that a polynomial-time approximation of $4/3 - 1/(3m)$ can be achieved using the \emph{Longest Processing Time First (LPT)} algorithm \cite{Graham69Schedulling}. 

Strategic variants of this problem, known as the weighted singleton congestion game \cite{Milchtaich1996WeightedSingletonCongestionGame} or the selfish load balancing problem \cite{Vöcking2007SelfishLoadBalancing}, some equivalent to
our study of credibility, 
were studied later.
These build on important properties of congestion games, such as the finite improvement property, existence of pure-strategy Nash equilibria, or the correspondence with potential games established in \cite{Rosenthal1973CongestionGames, Monderer1996PotentialGames}. 
The LPT algorithm was found to provide a pure-strategy Nash equilibrium 
\cite{Fotakis2002LPT_NE}. 
Subsequent work on the price of anarchy of atomic congestion games (e.g., \cite{Koutsoupias2009WorstCaseEquilibria,Christodoulou05PoACongestionGames,Christodoulou2005CorrelatedEquilibriaPoA,Awerbuch2005PoA,Aland2011PoA,Bhawalkar2014PoA}) studied the efficiency gap between the worst equilibrium and the efficient outcome. 

This paper is most closely related to recent works on communication partition \cite{Lee2025CommunicationPartitions, Lee2025PartitionEquilibria}, where players are partitioned into coalitions and communicate locally within their coalition. 
They play a correlated action profile (assignment) if a credible, efficient, and fair agreement can be reached within each coalition. 
We expand on this work in two important dimensions. 
First we relax the arguably very restrictive envy-freeness property by considering a hierarchy of weaker notions of fairness; second, rather than characterizing the structure of desirable outcomes, 
we study the computational complexity of finding an assignment for a coalition of players under all combinations of such axiomatic properties.

\begin{table}[th!]
    \centering
    \caption{Summary of the satisfiability and complexity for combinations of various fairness and credibility conditions.
    For definitions see \Cref{sec:preliminaries}.
    In the third and fourth column, we either give a polynomial bound on the running time for computing a satisfying assignment with or without a makespan constraint, or state NP-completeness of deciding on existence.\label{tab:summary}}
    \resizebox{1\columnwidth}{!}{
    \begin{tabular}{ccccc}
    \toprule
    Property $P$        & Exists     & Unique    & \Psat  & \Pmksp   \\
    \midrule
    $\top$              & Y             & N             & $O(n)$            & NP-complete \cite{Garey1979IMS_NP-complete} \\            
    $Eq$                & Y             & N             & $O(n)$            & NP-complete    \\
    $Cr$                & Y             & N             & $O(n \log n)$ \cite{Fotakis2002LPT_NE}    & NP-complete    \\
    $Eq \land Cr$       & N             & N             & NP-complete       & NP-complete    \\
    \midrule
    $EF$                & N             & N             & $O(n \sqrt n)$    & $O(n \sqrt n)$ \\
    $EF \land Cr$       & N             & Y             & $O(n \sqrt n)$    & $O(n \sqrt n)$ \\
    \midrule
    $WOE$               & Y             & N             & $O(n)$            & $O(n \log n)$ \\
    $WOE \land Eq \equiv OE$  & Y       & N             & $O(n)$            & $O(n^2)$ \\
    $WOE \land Cr$      & N             & Y             & $O(n \log n)$     & $O(n \log n)$ \\
    $WOE \land Eq \land Cr \equiv OE \land Cr$ & N & Y  & $O(n \log n)$     & $O(n \log n)$ \\
    \midrule
    $SM$                & Y             & N             & $O(n \log n)$     & $O(n \log n)$ \\
    $SM \land Eq$       & Y             & N             & $O(n \log n)$     & $O(n \log n)$ \\
    $SM \land Cr$       & N             & Y             & $O(n \log n)$     & $O(n \log n)$ \\
    $SM \land Eq \land Cr$   & N        & Y             & $O(n \log n)$     & $O(n \log n)$ \\
    \midrule
    $WM$                & Y             & N             & $O(n)$            & NP-complete   \\
    $WM \land Eq \equiv M$       & Y    & N             & $O(n)$            & NP-complete   \\
    $WM \land Cr$       & N             & N             & NP-complete       & NP-complete   \\
    $WM \land Eq \land Cr \equiv M \land Cr$ & N & N    & NP-complete       & NP-complete   \\
    \bottomrule
    \end{tabular}
    } 
\end{table}

\subsection{Our Contribution}

We answer the following four questions for each (single or composite) property $P$, a summary of which is provided in \Cref{tab:summary}.
\begin{itemize}
    \item Existence: Does a $P$-satisfying assignment always exist? 
    \item Uniqueness: If a $P$-satisfying assignment exists, is it unique?
    \item \Psat: What is the computational complexity of finding a $P$-satisfying assignment?
    \item \Pmksp: What is the computational complexity of finding a $P$-satisfying assignment of smallest makespan?
\end{itemize}

We find that, for certain classes of this decision problem, it is sufficient to search for a `contiguous' assignment which limits the search space and avoids the combinatorial blow-up. However, especially when this is not possible, we provide a reduction to show that they are indeed NP-complete.

For the remainder of this paper, we formally define the problem, our notation, and the various properties of interest in \Cref{sec:preliminaries}. We present tractable results (computable in polynomial time) in \Cref{sec:tractable}, and intractable results (NP-completeness) in \Cref{sec:intractable}. We conclude in \Cref{sec:conclusion}. 
To keep the paper concise, formal proofs are deferred to \Cref{sec:deferred_proofs}. 
Our main algorithm and all its subroutines are stated (or restated) in \Cref{sec:algorithm}. Results concerning envy-freeness are presented in \Cref{sec:EF}. 


\section{Preliminaries} \label{sec:preliminaries}


\subsection{Notation}

For positive integers $k$ and $\ell$, we denote $[k] :=\{1,\dots,k\}$ and $[k:\ell] := \{k,\dots, \ell\}$.
Given a vector $\textbf{v}$, we denote by $|\textbf{v}|$ its length.
We have a set $[n]$ of $n$ weighted players, and a set $[m]$ of $m$ identical resources. 
The weights of players are represented by a vector $\ww = (w_i)_{i\in [n]} \in \mathbb{Z}_{>0}^n$ of positive integers.
We denote by $\supp(\ww)$ 
the set of unique weights in $\ww$, and for $w \in \supp(\ww)$, we set $\#(w) := |\{i\in [n] \colon w_i = w\}|$ the number of players of weight $w$ with respect to $\ww$. An assignment is an n-tuple $\ba = (a_1, \dots, a_n)$ where $a_i \in [m]$ specifies the resource that is assigned to player~$i$. 

We can also express $\ba$ in terms of the \emph{load distribution} on each resource. 
For example, consider the instance $\ww = (4, 3, 3, 1)$, $m = 3$, and assignment $\ba = (1,1,2,2)$, leaving the third resource unused. 
The load distribution can be written as $l(\ba) = \{\{4, 3\} , \{3, 1\} ,\varnothing \}$, a multiset whose components are multisets. 
This representation groups equivalent assignments, only differing between players of same weight, into the same load expression. 
Two assignment are equivalent if they have the same load distribution, i.e., $\ba \equiv \ba' \iff l(\ba) = l(\ba')$.

Given an assignment $\ba$, the load on some resource $x \in [m]$ is the sum of all the players' weights assigned to it, denoted as $v_x(\ba) := \sum_{i:a_i = x} w_i$, and is the cost incurred by each player assigned to $x$.\footnote{Equivalently, this is a weighted singleton congestion game with a linear cost function.} The makespan of the assignment is the maximum load on any machine, i.e. $c(\ba) := \max_{x \in [m]} v_x(a)$. For brevity and consistency, we use the variables $i, j, k \in [n]$ for player indices or counts, and the variables $x, y, z \in [m]$ for resource indices or counts, throughout this paper. Additionally, $W := \sum_i w_i$ denotes the total weight of all players, $\ww[i:j] := (w_i, ..., w_j)$ denotes a contiguous list of sorted weights, and $W[i:j] := \sum_{k=i}^j w_k$ denotes the sum of weights of such a contiguous set of players.

\subsection{Problem Definition}

The problem we study can be formulated in a number of equivalent ways. We define it as a constrained identical machine scheduling problem to emphasize its connection with the classical identical machine scheduling problem, particularly on the computational complexity of finding an optimal assignment.



Given a property $P$ (this can be a conjunction of properties), our main problem is to minimize makespan subject to satisfying this property.

\problem{\PmkspT}{
A weight vector $\ww = (w_1, ..., w_n)$ of positive integers, a positive integer $m$, makespan threshold $t$.
}{
Find a $P$-satisfying assignment of makespan at most $t$, or decide that no such assignment exists.
}

We also investigate the problem of satisfying $P$ in isolation. 
It can be formulated as a special case of \Pmksp{} with $t = \infty$. 

\problem{\PsatT}{
A weight vector $\ww = (w_1, ..., w_n)$ of positive integers and a positive integer $m$.
}{
Find a $P$-satisfying assignment, or decide that no such assignment exists.
}

Throughout this paper, we assume that $n > m > 1$ (i.e., there are more players than resources, and more than one resource), and that $\max_i w_i \leq t$. 
The case $n \leq m$ is rather trivial: an assignment where each player is assigned a different resource satisfies all of the conditions considered in this paper, so we can answer \Pmksp{} by simply checking whether the largest weight in $\ww$ exceeds the makespan threshold $t$. 
Moreover, the case $m = 1$ allows only for a single assignment (all players assigned to the same resource), which can be checked to solve \Pmksp. 
Finally, the case $\max_i w_i > t$ clearly has no solution. 
  
\subsection{Properties and Their Relations}

The properties that we consider in this paper are defined as follows. 
Single properties can be combined to form composite properties, e.g., $\ba$ satisfies $WOE \land Eq$ means $\ba$ satisfies $WOE$ and $\ba$ satisfies $Eq$.
Recall that $v_{a_i}$ is the load on the resource assigned to player~$i$.

\begin{definition} [Properties]
    Given an assignment $\ba$, we say that $\ba$ satisfies the following property if for all $i, j \in [n], x \in [m]$:

    \textbf{Credibility ($\boldsymbol{Cr}$)}: $v_{a_i} \leq v_x + w_i$.

    \textbf{Equality ($\boldsymbol{Eq}$)}: $w_i = w_j \implies v_{a_i} = v_{a_j}$.    
    
    
    \textbf{Envy-freeness ($\boldsymbol{EF}$)}: $v_{a_i} - w_i \leq v_{a_j} - w_j$.    
    
    \textbf{Weak ordered envy-freeness ($\boldsymbol{WOE}$)}: $w_i < w_j, a_i \neq a_j  \implies v_{a_i} - w_i \leq v_{a_j} - w_j$. 

    \textbf{Ordered envy-freeness ($\boldsymbol{OE}$)}: $w_i \leq w_j, a_i \neq a_j \implies v_{a_i} - w_i \leq v_{a_j} - w_j$.
    
    \textbf{Strong monotonicity ($\boldsymbol{SM}$)}: $w_i < w_j \implies v_{a_i} < v_{a_j}$.
    
    \textbf{Weak monotonicity ($\boldsymbol{WM}$)}: $w_i < w_j \implies v_{a_i} \leq v_{a_j}$.

    \textbf{Monotonicity ($\boldsymbol{M}$)}: $w_i \leq w_j \implies v_{a_i} \leq v_{a_j}$.
\end{definition}

We speak of a $P$-assignment to refer to an assignment satisfying $P$.
We use the symbol $\top$ to denote the empty disjunction of properties, i.e., there are no constraints.
$Cr$ is the property that no player can benefit from a unilateral deviation, or equivalently the agreement is a pure-strategy Nash equilibrium, which ensures that the non-binding agreement is honoured by the players. 
$EF$ is a strong fairness property that requires no player to want to switch their assignment with another. 
$WOE$ is a weaker version of $EF$, which only requires the lighter player of any pair of players to not envy the assignment of the heavier player. 
The condition applies only to pairs of players that are assigned different resources, since it would make no sense for a player to envy another player assigned to the same resource. $SM$ is another fairness condition that says a heavier player should incur more cost than a lighter player, while $WM$ says that a heavier player should incur no less cost than a lighter player. 
Since $WOE$, $SM$, and $WM$ do not impose any fairness conditions on players of the same weight, we also introduce $Eq$, which requires players of the same weight to incur the same cost. Finally, $OE$ and $M$ combine the $Eq$ condition with $WOE$ and $WM$ respectively, as we show in \Cref{prop:relations}.

\begin{remark}[Other properties]
    \cite{Lee2025CommunicationPartitions, Lee2025PartitionEquilibria} additionally consider Pareto optimality; however, since it is implied by the stronger condition of credibility in our model, and is trivially satisfied by any assignment that uses all resources (in the case $n > m$), we do not separately consider it in this paper.
    Also notice that we do not define `Strong ordered envy-freeness' ($SOE$), which would logically be the condition $v_{a_i} - w_i < v_{a_j} - w_j$ for all $i, j : w_i < w_j, a_i \neq a_j$. This is due to the fact that this is an unreasonably strong condition which cannot be satisfied even in trivial cases. 
    For example, no $SOE$-assignment exists for the instance $\ww = (2, 3), m = 2$; even the assignment where $l(\ba) = \{\{3\},\{2\}\}$ leads player of weight~$2$ to strongly envy the player of weight~$3$.
\end{remark}

\begin{remark}[Equivalent problems]
    Notice that the \Pmksp[\top] problem, where no constraint is applied, is equivalent to the classical identical machine scheduling problem \cite{Graham69Schedulling}. The \Psat[Cr] problem is equivalent to finding a pure-strategy Nash equilibrium in the singleton congestion game \cite{Christodoulou05PoACongestionGames}, 
    and the \Pmksp[Cr] problem can be used to find the socially optimal pure-strategy Nash equilibrium via a binary search for the optimal $t$.
\end{remark}

The relations between these properties are set out in \Cref{prop:relations}. 
\Cref{fig:relations} illustrates them with a Venn diagram.

\begin{restatable}{proposition}{PropRelations}[Property relations]\label{prop:relations}
    The following relations hold: 
    \begin{enumerate}
        \item \textbf{$EF \implies Eq$}
        \item \textbf{$EF \implies WOE \implies WM$}
        \item \textbf{$EF \implies SM \implies WM$}
        \item \textbf{$WOE \land Eq \iff OE$}
        \item \textbf{$WM \land Eq \iff M$}
    \end{enumerate}
\end{restatable}

It might seem natural that $WOE \implies SM$, but this is not true, due to the fact that $WOE$ does not consider envy between players assigned to the same resource. Examples of assignments satisfying various properties are shown in \Cref{fig:assignment_examples}. 

\begin{figure}[h]
    \centering
    \includesvg[width=0.7\linewidth]{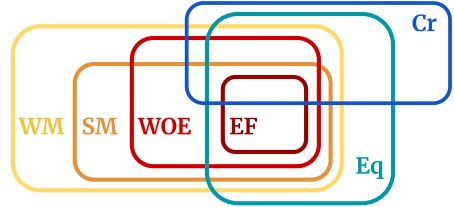}
    \caption{Venn diagram of relations between properties.}
    \label{fig:relations}
\end{figure}

\begin{figure*}
    \centering
    \adjustbox{trim=0cm 0.2cm 5.5cm 6cm}{\includesvg[width=1.3\linewidth]{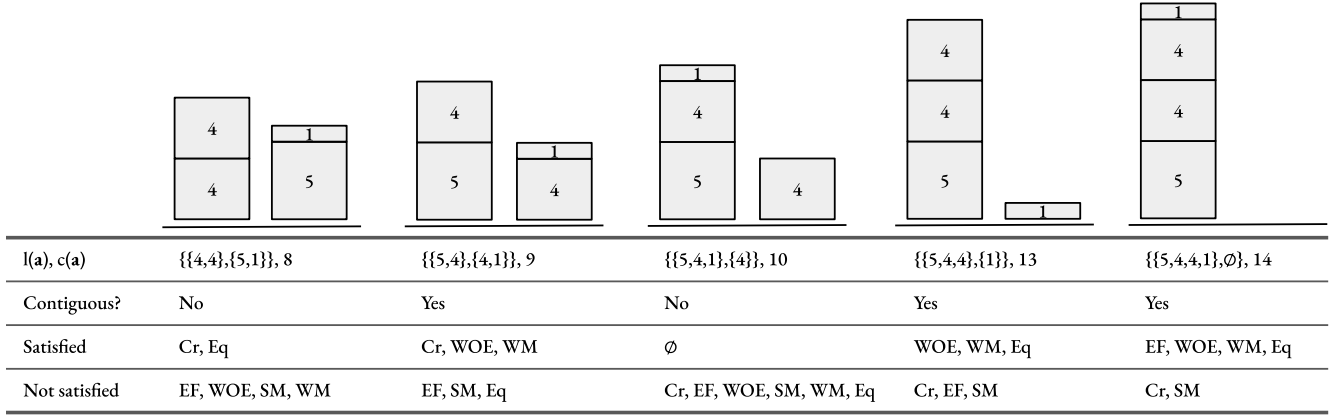}}
    \caption{Examples of assignments for an instance where $\ww = (5, 4, 4, 1)$, $m = 2$.}
    \label{fig:assignment_examples}
\end{figure*}

\subsection{Contiguous Assignment}

An important property that the satisfying assignments of many (conjunctions of) properties 
have is \emph{contiguity} as defined next.
This property allows us to quickly search for a satisfying assignment in its \emph{canonical ordering}.  


\begin{definition} [Contiguous assignment]
    An assignment is \emph{contiguous} if under a non-increasing ordering of player weights where players are indexed such that $i < j \implies w_i \geq w_j$, we have $(i < j) \land (a_i = a_j) \implies a_i = a_k \ \forall k \in [i:j]$. 
\end{definition}

\begin{definition} [Canonical ordering]
    Assume we are given a contiguous assignment where players are indexed in a non-increasing order satisfying the contiguity definition.
    Then, we can further index the resources in a non-increasing order such that, for all $i, j \in [n], w_i \geq w_j \implies a_i \leq a_j$. 
    We call this ordering of players and resources the \emph{canonical ordering} of the contiguous assignment.
\end{definition}

Thus, under the canonical ordering of a contiguous assignment, the first contiguous set of players are assigned to the first resource, the second contiguous set of players to the second resource, etc. The search for a contiguous assignment only has to find the right `cuts' to make in order to assign players to resources, which significantly reduces the computational complexity of the problem. 

\begin{example}
    Consider the load multiset $l(\ba) = \{\{2, 3\} , \{3, 3\} ,\{1, 2\}\}$. 
    This is a contiguous assignment since we can order the player weights as $(3, 3, 3, 2, 2, 1)$.
    We can then index the resources to express the assignment as $\{\{3, 3\} , \{3, 2\} ,\{2, 1\}\}$ in its canonical ordering.
\end{example}




\section{Tractable Properties} \label{sec:tractable}



    
    
    

In this section, we aim at algorithmic feasibility.
We present a general algorithmic framework based on \Cref{alg:generic}, a polynomial-time algorithm, which we term \emph{Sequential Contiguous Assignment (SCA)}.
It exploits the fact that a $P$-assignment is necessarily contiguous for a wide range of properties. 
The algorithm does so by searching for a contiguous assignment in its canonical ordering. 
We start by setting out the general, modular structure of the algorithm, and then specify the subroutines corresponding to various properties $P$ in subsequent sections.

\begin{algorithm} 
    \caption{Sequential Contiguous Assignment (SCA)\label{alg:generic}}
    \begin{algorithmic}[1]
    \Statex \textbf{Input:} A weight vector $\ww = (w_1, w_2, \dots, w_n)$, a number of resources $m$, a property $P$, and a makespan threshold~$t$.
    \Statex \textbf{Output:} A contiguous $P$-assignment of makespan at most $t$, if it exists; \outpt{None} otherwise.

    \vspace{0.5em}
    \State Sort and index $\ww$ in non-increasing order of weights.

    \vspace{0.5em}
    \State $k \gets \textsc{First\_k}(P, t, \ww, m)$
    \State $a_j \gets 1$ for all $j \in [1:k]$ \Comment{Assign players to 1st resource}
    \State $x \gets 2$ \Comment{Define resource counter}
    \State $i \gets k+1$ \Comment{Define player counter}
    \State $t \gets W[1:k]$ \Comment{Update load threshold}
    
    \vspace{0.5em}
    \While{$x \leq m$ and $i \leq n$}
        \State Compute relevant input $z$
        \State $k \gets \textsc{Subsequent\_k}(P, t, 
        \ww[i:n], m-x+1, z)$
        \State $a_j \gets x$ for all $j \in [i,i+k]$ \Comment{Assign players to resource $x$}
        \State $x \gets x+1$ \Comment{Update resource counter}
        \State $i \gets i+k$ \Comment{Update player counter}
        \State $t \gets W[i: i+k-1]$ \Comment{Update load threshold}
    \EndWhile

    \vspace{0.5em}
    \If{$i < n$} \Comment{Some players remain unassigned}
        \State \Return \outpt{None} 

    \vspace{0.5em}
    \Else \Comment{All players assigned to resources}
        \State $a \gets (a_1, ..., a_n)$
        \If{\textsc{Satisfies}$(\ba, P)$} 
        \Comment{Test $\ba$ for $P$}
            \State \Return $a = (a_1, ..., a_n)$
        \Else
            \State \Return \outpt{None}
        \EndIf
    \EndIf
    \end{algorithmic}
\end{algorithm}

Players are first sorted in a non-increasing order. Then the algorithm iteratively finds the next number of players~$k$ to assign to a new resource.
It proceeds until either a satisfying assignment is found, or it runs out of resources before all players are assigned, in which case no satisfying assignment is found. 
A final check is conducted for some cases of $P$ to ensure that the candidate assignment does indeed satisfy~$P$. 
The subroutines \textsc{First\_k}, \textsc{Subsequent\_k} and \textsc{Satisfies} differ for different problems, and are specified in the dedicated sections.
Note that a subroutine might also return \outpt{None}, which causes termination of the algorithm.

We make a few general observations about the algorithm: First, the makespan threshold~$t$ is updated after each iteration to be the load of the resource that was just assigned to ensure a non-increasing ordering of load, which is necessary for certain fairness conditions. 

Second, after the assignment of the first $k$ players to the first resource, \textsc{Subsequent\_k} is repeatedly called with the remaining players and resources until all players or resources are assigned. An additional argument $z$ is computed and passed to the \textsc{Subsequent\_k} subroutine; $z$ contains information about the players assigned to the previous resource (e.g., maximum weight, minimum weight, number of players), which is used by \textsc{Subsequent\_k} to ensure the satisfaction of $P$ in the current assignment.

Finally, the algorithm allows for only a single value of $k$ to be returned by the \textsc{First\_k} and \textsc{Subsequent\_k} subroutines at each iteration. This is sufficient as the algorithm searches for the assignment that is `closest' to satisfying $P$, so that if no assignment is found, then no satisfying assignment exists. 

\subsection{Simple Satisfiability} 

For a number of simple properties, one can always find a $P$-assignment (rendering the decision problem \Psat{} trivial) 
in polynomial time, while \Pmksp{} is NP-complete.
We now discuss such properties $P$, 
and address the corresponding \Pmksp{} problems in \Cref{sec:intractable}.

\subsubsection{\Psat[\top]}

Any assignment satisfies $\top$ when there is no makespan threshold, and thus its existence is guaranteed, but not uniqueness.

\subsubsection{\Psat[Eq]}

Consider an assignment where all players are assigned to the same resource. 
This satisfies $Eq$ since all players incur the same cost $W$ (which is the load of the resource to which they are assigned). 
Its existence is thus guaranteed, but not uniqueness.

\subsubsection{\Psat[Cr]}


The Longest Processing Time algorithm \cite{Graham69Schedulling} is known to find a $Cr$-assignment 
in time $O(n\log n)$ \cite{Fotakis2002LPT_NE}. 
It is a greedy algorithm that first sorts players by non-increasing weight, and then sequentially assigns them to the resource with the least current load. This result is included in \Cref{tab:summary}.

\subsection{Envy-Freeness}

Envy-freeness as a fairness property is extensively discussed in \cite{Lee2025CommunicationPartitions, Lee2025PartitionEquilibria} and not the primary focus of this paper. Nevertheless, we provide an algorithm for solving \Pmksp[EF] and \Pmksp[EF \land Cr] in time $O(n \sqrt n)$ in \Cref{sec:EF}.

\subsection{Weak Ordered Envy-Freeness} 

Since $WOE$ does not place any restrictions on players assigned to the same resource (i.e., a player cannot be envious of another that is assigned to the same resource), assigning all players to a single resource satisfies $WOE$. 
The existence of a $WOE$-assignment is, therefore, guaranteed, and it can be returned in time $O(n)$. 
However, it may not be unique. 
For example, the assignments $\ba$ and $\ba'$ where $l(\ba) = \{\{2, 2\}, \{2, 1\}\}$ and $l(\ba') = \{\{2,2,2\},\{1\}\}$ both satisfy $WOE$. 
Moreover, we can make a stronger statement about $WOE$-assignments.

\begin{restatable}{lemma}{LemWOEContiguous} \label{lem:WOE_contiguous}
    Every $WOE$-assignment is contiguous.
\end{restatable}

Hence, it is sufficient to consider only contiguous assignments to answer \Pmksp{} if $P$ contains $WOE$.

In addition, we introduce a lemma which says that under the canonical ordering, if $WOE$ holds between players in adjacent resources, then $WOE$ holds globally.

\begin{restatable}{lemma}{LemWOEGlobal} \label{lem:local_to_global_WOE}
    Let $\ba$ be a contiguous assignment in its canonical ordering. If the $WOE$ condition holds between players assigned to any two adjacent resources, 
    then $\ba$ satisfies $WOE$.
\end{restatable}

\Cref{lem:WOE_contiguous} suggests that \Cref{alg:generic} can be used to solve \Pmksp{} whenever WOE is entailed by $P$.
Hence, we obtain a correct algorithm by specifying the subroutines.
While we defer a formal proof of correctness to the appendix, we subsequently discuss the subroutines and their intuition.
Note that the formal proof uses \Cref{lem:local_to_global_WOE} to prove WOE.

\begin{restatable}{theorem}{TheoremWOE}\label{thm:WOE_algo_correctness}  
    \Pmksp{} can be solved in time $O(n \log n)$ for $P \in \{WOE, WOE \land Cr, WOE \land Eq \land Cr\}$, and in time $O(n^2)$ for $P = WOE \land Eq$. 
\end{restatable}

\subsubsection{\Pmksp[WOE]}

We define three subroutines that are called by \Cref{alg:generic} to solve the \Pmksp[WOE] problem.
The idea is to assign as many players (in non-increasing weight order) to a resource as possible.
This reduces the propensity for a lighter player that is yet to be assigned to envy a heavier player that is assigned to a high-load resource, and minimizes the total number of resources required to assign all players.


\begin{restatable}{algorithmic}{WOEFirstk}
    \Procedure{First\_k}{$P = WOE, t, \ww, m$}
        \State \Return $\max \{ i\in [n] \colon W[1:i] \leq t\}$
    \EndProcedure
\end{restatable}  

This subroutine assigns the maximum number of players to the first resource without exceeding the makespan threshold~$t$. 
Notice that for \Psat[WOE] where $t = \infty$, the subroutine returns $k = n$, i.e., all players are assigned to the first resource.

For the assignment of players to subsequent resources, let $z$ be the largest weight assigned to the previous resource. 
This information is used to determine the largest number of players that can be assigned to the next resource without triggering envy, i.e., we need to ensure the that load on the current resource minus its lightest player does not exceed the load on the previous resource minus its heaviest player. 
The last $\min$ operation ensures that $k$ does not exceed the number of remaining players.


\begin{restatable}{algorithmic}{WOESubsequentk}
    \Procedure{Subsequent\_k}{$P = WOE, t, \ww, m, z$}
        \State Set $\hat i = \max \{ i \colon W[1:i] \leq t - z \}$
        \State \Return $\min(\hat i + 1, |\ww|)$
    \EndProcedure
\end{restatable}

Notice that, by construction, the $WOE$ condition is satisfied between players assigned to any two adjacent resources. By \Cref{lem:local_to_global_WOE}, this means that the candidate assignment satisfies $WOE$. So no further check for $WOE$-satisfaction is required and the \textsc{Satisfies} subroutine always returns \outpt{True}.




\subsubsection{\Pmksp[WOE\land Eq]}

We begin with a structural proposition about $WOE \land Eq$-assignments, which says that if players of the same weight are assigned to different resources, then it must be the case that all players sharing those resources are of the same weight.

\begin{restatable}{lemma}{LemWOEEq} \label{lem:WOE_Eq}
    Let $\ba$ be a $WOE \land Eq$-assignment. Assume that $i, j \in [n]$ with $w_i = w_j$ and $a_i \neq a_j$. Then $w_k = w_j$ for all $k \in [n]$ with $a_k = a_i$ or $a_k = a_j$. 
\end{restatable}

We thus need to modify our previous $WOE$ subroutines to ensure $Eq$-satisfaction by changing the way $k$ is determined in \textsc{First\_k} and \textsc{Subsequent\_k}: 

\begin{itemize}
    \item If not all players assigned to the same resource have the same weight, then the next player to be assigned to the next resource (the $(k+1)$st player) cannot share the same weight as the last player of the current assignment (the $k$th player).
    \item If all players assigned to the same resource have the same weight, then all players of such weight must be evenly distributed among some number of resources so they incur the same cost. I.e., if $k$ players of the same weight $w_i$ are assigned to some resource, then it must be the case that $k \mid \#(w_i)$ and all assigned resources have the same load.
\end{itemize}

The subroutine \textsc{First\_k} for \Pmksp[WOE\land Eq] is as follows:

\begin{restatable}{algorithmic}{WOEEqFirstk}
    \Procedure{First\_k}{$P = WOE \land Eq, t, \ww, m$}
        \State Set $n = |\ww|$, $w_{n+1} = 0$ 
        \If{ $\{ i \in [n]\colon (W[1:i] \leq t) \land (w_{i+1} \neq w_i)  \}\neq \emptyset$}
            \State \Return $\max \{ i \in [n]\colon (W[1:i] \leq t) \land (w_{i+1} \neq w_i)  \}$
        \Else{}
            \State \Return $\max \{ i\in [n]\colon (i \mid \#(w_1)) \land (i \cdot w_1 \leq t) \}$
        \EndIf
    \EndProcedure
\end{restatable}

For \textsc{Subsequent\_k}, we need a number of additional information to ensure $WOE \land Eq$-satisfaction. 
Let $z$ be the tuple $(\underline{w}, \overline{w}, k')$, representing the weight of the last player, the weight of the first player, and the number of players assigned to the previous resource, respectively.
$\underline{w}$ and $k'$ are used to ensure satisfaction of $Eq$ by checking whether the current first player shares the same weight as the last player assigned to the previous resource; if so, then we simply need to assign $k'$ players (of the same weight) to the current resource. $\overline{w}$ is used to ensure $WOE$ satisfaction, as we have seen before.

\begin{restatable}{algorithmic}{WOEEqSubsequentk}
    \Procedure{Subsequent\_k}{$P = WOE \land Eq, t, \ww, m, z$}
        \State Let $z = (\overline{w}, \underline{w}, k')$
        \If { $w_1 = \underline{w}$} \Comment{Continuation of players of equal weight}
            \State \Return $k'$
        \EndIf
    
        \State Set $n = |\ww|$, $b = t - \overline{w}$, $w_{n+1} = 0$.
        \If{ $\{ i \in [n]\colon (W[1:i-1] \leq b) \land (w_{i+1} \neq w_i)  \}\neq \emptyset$}
            \State \Return $\max \{ i \in [n]\colon (W[1:i-1] \leq b) \land (w_{i+1} \neq w_i)  \}$
        \Else{}
            \State \Return $\max \{ i\in [n]\colon (i \mid \#(w_1)) \land ((i-1) \cdot w_1 \leq b) \}$
        \EndIf
    \EndProcedure
\end{restatable}


Again, no further checks whether $WOE \land Eq$ is satisfied is required, so the \textsc{Satisfies} subroutine always returns \outpt{True}.


\subsubsection{\Pmksp[WOE\land Cr]}


While the $WOE$ condition implicitly requires that resources assigned later have no greater load than the resources before it, the $Cr$ condition requires that such differences are not so large that players assigned to earlier resources could benefit from deviating to later resources. 
The idea is that, in each iteration, the subroutines now assign a number of players $k$ to a resource to ensure that the average load of the remaining resources, $\frac{W - W[1:k]}{m-1}$, can satisfy $WOE \land Cr$, i.e.,


\begin{equation} \label{eqn:WOE_Cr_bounds_on_remainders}
    W[1:k-1] \leq \frac{W - W[1:k]}{m-1} \leq W[1:k]
\end{equation}

The first inequality is a necessary (but not sufficient) condition to satisfy $Cr$: the average load on the remaining resources must be sufficiently large to prevent player~$k$ from deviating to the least-loaded resource. We need only satisfy this for player~$k$ (and not others assigned to the same resource) to produce the largest lower-bound.

The second inequality is a necessary (but not sufficient) condition to satisfy $WOE$: it ensures that the average load on the remaining resources is sufficiently small to ensure a canonical ordering of contiguous assignment (\Cref{lem:WOE_contiguous}).
%
%
It's easy to see that at most one value of~$k$ satisfies these bounds:

\begin{restatable}{lemma}{LemUniquek} \label{lem:unique_k}
    At most one value of $k$ satisfies \Cref{eqn:WOE_Cr_bounds_on_remainders}.
\end{restatable}

As we apply this argument to every iteration, we must reach a unique canonical ordering of resources if one exists.

\begin{corollary}[Uniqueness] \label{prop:WOE_Cr_uniqueness}
    If $\ba$ and $\ba'$ satisfy $WOE \land Cr$, then $\ba$ and $\ba'$ are identical under their canonical orderings.
\end{corollary}

The example below demonstrates that $k$ may not exist.

\begin{example} 
    Consider the instance where $\ww = (3, 2, 2)$ and $m = 2$. 
    If $k = 1$, then we get $(W - W[1 \colon 1])/(m-1) = 4$, which exceeds the upper-bound $W[1 \colon 1] = 3$. If $k = 2$, then $(W - W[1 \colon2])/(m-1) = 2$, which falls below the lower-bound $W[1 \colon 1] = 3$. 
\end{example}

Explicitly, the assignment of players to resources is governed by the following subroutines:


\begin{restatable}{algorithmic}{WOECrFirstk}
    \Procedure{First\_k}{$P = WOE \land Cr, t, \ww, m$}
        \If{ $\exists k\colon W[1:k-1] \leq \frac{W - W[1:k]}{m-1} \leq W[1:k] \leq t $}
            \State \Return $k$    
        \Else{}
            \Return  \outpt{None}
        \EndIf
    \EndProcedure
\end{restatable}

\begin{restatable}{algorithmic}{WOECrSubsequentk}
    \Procedure{Subsequent\_k}{$P = WOE \land Cr, t, \ww, m, z$}
    \If { $m = 1$} \Comment{Final resource}
        \State \Return $|\ww|$
    \EndIf
    \If { $\exists k\colon W[1:k-1] \leq \frac{W - W[1:k]}{m-1} \leq W[1:k] \leq t $}
        \State \Return $k$
    \Else{}
        \Return  \outpt{None}
    \EndIf
    \EndProcedure
\end{restatable}

Finally, even if the iterations of assigning players to resources result in a unique candidate assignment $\ba$, $\ba$ has only satisfied the necessary, but not sufficient, conditions of $WOE \land Cr$. 
Thus a final test is required to check that the candidate assignment satisfies $WOE \land Cr$. 
For $WOE$, we need only check that the lightest player of each resource assignment does not envy the heaviest player assigned to the previous resource, and there are at most $m-1 < n$ resources to check; for $Cr$, since resource loads are ordered non-increasingly, we need only check that the heaviest player of each resource does not benefit from deviating to the last and least-loaded resource, for which there are again $m-1 < n$ resources to check.\footnote{For simplicity we check that each player has no incentive to deviate to the least-loaded resource in \Cref{sec:algorithm}.} These checks can thus be done in time $O(n)$.

\subsubsection{\Pmksp[WOE \land Eq \land Cr]}

Since the previous algorithm for \Pmksp[WOE \land Cr] produces only a single feasible candidate, we need only conduct an additional check on whether this candidate satisfies $Eq$ in the \textsc{Satisfies} subroutine 
to answer this case. 
The subroutines \textsc{First\_k} and \textsc{Subsequent\_k} remain unchanged.






\subsubsection{Satisfiability Problems}

Notice that for \Psat[WOE] and \Psat[WOE \land Eq], assigning all players to the first resource is always a solution, which can be returned in time $O(n)$.  Also, because there is at most one feasible assignment for the \Pmksp[WOE \land Cr] and \Pmksp[WOE \land Eq \land Cr] problems, a solution is unique if it exists. 
\Cref{thm:WOE_algo_correctness} implies that the corresponding satisfaction problem (i.e., with $t = \infty$) can be solved in $O(n \log n)$.

\begin{corollary}
    \Psat[WOE] and \Psat[WOE \land Eq] can be solved in time $O(n)$. 
    \Psat[WOE \land Cr] and \Psat[WOE \land Eq \land Cr] can be solved in time $O(n \log n)$. 
\end{corollary}


\subsection{Strong Monotonicity}

We first notice that an $SM$-assignment is always contiguous: 
if it was not, then under any non-increasing ordering of players and resources we can find some $i, j \in [n]: w_i < w_j$ and $a_i < a_j$, which contradicts $SM$.

\begin{restatable}{lemma}{LemSMContiguous} \label{lem:SM_contiguous}
    Every $SM$-assignment is contiguous. 
\end{restatable}

Another observation we can make is that in an $SM$-assignment, all players sharing the same resource must have the same weight; otherwise these players would incur the same cost, which fails the $SM$ condition for the lighter players.

\begin{lemma} \label{lem:SM_no_sharing_between_diff_weights}
    Let $\ba$ be an $SM$-assignment.
    If $a_i = a_j$ for players $i, j\in [n]$, then $w_i = w_j$. 
\end{lemma}

In the remainder of this section, we discuss and present the relevant subroutines of the SCA algorithm for each property containing $SM$, to support the following theorem. 

\begin{restatable}{theorem}{TheoremSM}\label{thm:SM_algo_correctness}  
    \Pmksp{} and \Psat{} can be solved in time $O(n \log n)$ for $P \in \{SM, SM \land Eq, SM \land Cr, SM \land Eq\land Cr\}$. 
\end{restatable}

\subsubsection{\Pmksp[SM]}

Given \Cref{lem:SM_contiguous,lem:SM_no_sharing_between_diff_weights}, 
a natural way of searching for an assignment is to look for its canonical ordering, and sequentially assign all players of a given weight to as few resources as possible while ensuring $SM$-satisfaction. 
 
\begin{restatable}{algorithmic}{SMFirstk}
    \Procedure{First\_k}{$P = SM, t, \ww, m$}
        \If { $\#(w_1) \cdot w_1 \leq t$} 
            \State \Return  $\#(w_1)$    
        \Else{}
            \State Compute $\hat k = \max \{k : k \cdot w_1 \leq t\}$
            \State Set $r = \lceil \frac{\#(w_1)}{\hat k} \rceil$
            \State \Return  $\lceil \#(w_1)/r \rceil$
        \EndIf
    \EndProcedure
\end{restatable}

In \textsc{First\_k}, if $\#(w_1)$ is too large for all players of weight $w_1$ to be assigned to the first resource, then we first calculate the minimum number of resources required to accommodate all $\#(w_1)$ players, then distribute these players as evenly as possible, to ensure that the last resource to which players of weight $w_1$ are assigned has the highest possible load.
We do this to ensure that the new upper-bound on the load of the next resource is as high as possible.

\begin{restatable}{algorithmic}{SMSubsequentk}
    \Procedure{Subsequent\_k}{$P = SM, t, \ww, m, z$}
        \State Set $z$ to be the weight of players in previous resource
        \If { $z = w_1$} \Comment{Load must be $\leq t$}
            \If { $\#(w_1) \cdot w_1 \leq t$}
                \State \Return $\#(w_1)$    
            \Else{}
                \State Compute $\hat k = \max \{k : k \cdot w_1 \leq t\}$
                \State Set $r = \lceil \frac{\#(w_1)}{\hat k} \rceil$
                \State \Return $\lceil \#(w_1)/r \rceil$
            \EndIf
        \Else {} \Comment{Load must be $< t$}
            \If { $\#(w_1) \cdot w_1 < t$}
                \State \Return $\#(w_1)$    
            \Else{}
                \State Compute $\hat k = \max \{k : k \cdot w_1 < t\}$
                \State Set $r = \lceil \frac{\#(w_1)}{\hat k} \rceil$
                \State \Return $\lceil \#(w_1)/r \rceil$
            \EndIf
        \EndIf
    \EndProcedure
\end{restatable}

We use a similar process for \textsc{Subsequent\_k}, except that an additional check occurs on whether the current first player has the same weight as the players assigned to the previous resource. If so, then the load threshold is a weak inequality, since players of the same weight can incur the same cost. If not, then the load threshold is a strong inequality, to ensure that the players of the lighter weight incur strictly smaller cost than the heavier players.

Finally, the subroutines above ensure that there are no additional checks to be made, so the \textsc{Satisfies} subroutine simply returns true.

\subsubsection{\Pmksp[SM \land Eq]}

To search for a solution to \Pmksp[SM \land Eq], we need to additionally ensure that all players of the same weight incur exactly the same cost. This means that they must be distributed evenly across a number of resources (rather than the minimum number of resources), as shown in the following subroutines.

\begin{restatable}{algorithmic}{SMEqFirstk}
    \Procedure{First\_k}{$P = SM \land Eq, t, \ww, m$}
        \If { $\#(w_1) \cdot w_1 \leq t$}
            \State \Return $\#(w_1)$    
        \Else{}
            \State \Return  $\max \{k : (k \cdot w_1 \leq t) \land (k \mid \#(w_1)) \}$
        \EndIf
    \EndProcedure
\end{restatable}

\begin{restatable}{algorithmic}{SMEqSubsequentk}
    \Procedure{Subsequent\_k}{$P = SM \land Eq, t, \ww, m, z$}
        \State Let $z$ be the weight of players in previous resource
        \If { $z = w_1$} \Comment{Load must be $\leq t$}
            \If { $\#(w_1) \cdot w_1 \leq t$}
                \State \Return  $\#(w_1)$    
            \Else{}
                \State \Return $\max \{k : (k \cdot w_1 \leq t) \land (k \mid \#(w_i)) \}$
            \EndIf
        \Else{} \Comment{Load must be $< t$}
            \If { $\#(w_1) \cdot w_1 < t$}
                \State \Return $\#(w_1)$    
            \Else{}
                \State \Return $\max \{k : (k \cdot w_1 < t) \land (k \mid \#(w_i)) \}$
            \EndIf
        \EndIf
    \EndProcedure
\end{restatable}

No further checks are required, so the \textsc{Satisfies} subroutine always returns \outpt{True}. 


\subsubsection{\Pmksp[SM \land Cr]}

To satisfy $Cr$ we cannot simply assign the maximal number of players to each resource, as doing so might result in a resource being unused or having a low load, and thus make it beneficial for some players to deviate to this resource. 
We need to ensure that, when choosing the number of players $k$ to assign to each resource, that the remaining players can be distributed so that all subsequent resources have sufficiently high load. 

To ensure that an assignment satisfies $Cr$, we must impose the same bounds shown in \Cref{eqn:WOE_Cr_bounds_on_remainders} at each iteration of the assignment. 
As we have shown in \Cref{prop:WOE_Cr_uniqueness}, at most one value of $k$ exists at each iteration, so there is a unique candidate solution which we can find and check for the \Pmksp[SM \land Cr] problem. 




Since players sharing the same resource must be of the same weight, we can re-write \Cref{eqn:WOE_Cr_bounds_on_remainders} as $(k-1) w_1 \leq \frac{W - k w_1}{m-1} \leq k w_1$ where $k$ is the number of players of the weight $w_1$ assigned to the resource. 
The only integer value that $k$ can take is $\lceil \frac{W}{m w_1} \rceil$, if it exists, and we can check if it is less than or equal to $\#(w_1)$, the number of remaining players of weight $w_1$.\footnote{Rearranging the bounds gives $\frac{W}{m w_1} \leq k \leq \frac{W}{m w_1} + \frac{m-1}{m}$.}

Repeating this step for each resource eventually yields the result that either there is no satisfying assignment, or that all players are assigned, at which point we can check whether the candidate assignment satisfies $Cr$ overall.
This can be done by checking whether each player has the incentive to deviate to the least-loaded resource.\footnote{It is in fact sufficient to check that the heaviest player assigned to each resource has no incentive to deviate to the least-loaded resource.}
These steps are shown in the following subroutines to \Cref{alg:generic}.

\begin{restatable}{algorithmic}{SMCrFirstk}
    \Procedure{First\_k}{$P = SM \land Cr, t, \ww, m$}
        \If{ $\exists k : (k-1) w_1 \leq \frac{W - k w_1}{m-1} \leq k w_1 \leq t \land k \leq \#(w_1)$}
            \State \Return $k$      
        \Else{}
            \Return  \outpt{None}
        \EndIf
    \EndProcedure
\end{restatable}

\begin{restatable}{algorithmic}{SMCrSubsequentk}
    \Procedure{Subsequent\_k}{$P = SM \land Cr, t, \ww, m, z$}
    \State Let $z$ be the weight of players in previous resource
    \If { $z = w_1$} \Comment{Load must be $\leq t$}
        \If { $m = 1$ } \Comment{Final resource}
            \If {$w_1 = w_{|\ww|}$}
                \State \Return $|\ww|$
            \Else{}
                \Return \outpt{None}
            \EndIf
        \EndIf
        
        \If { $\exists k : (k-1) w_i \leq \frac{W - k w_i}{m-1} \leq k w_i \leq t \land k \leq \#(w_1)$}
            \State \Return $k$          
        \Else{}
            \Return  \outpt{None}
        \EndIf
    \Else {} \Comment{Load must be $< t$}
        \If { $m = 1 $} \Comment{Final resource}
            \If { $w_1 = w_{|\ww|} \land W[1:|\ww|] < t$}
                \State \Return $|\ww|$
            \Else {}
                \Return  \outpt{None}
            \EndIf
        \EndIf
        
        \If { $\exists k : (k-1) w_1 \leq \frac{W - k w_1}{m-1} \leq k w_1 < t \land k \leq \#(w_1)$}    
            \State \Return $k$      
        \EndIf
        \Return  \outpt{None}    
    \EndIf
    \EndProcedure
\end{restatable}

\begin{restatable}{algorithmic}{SMCrSatisfies}
    \Procedure{Satisfies}{$\ba, P = SM \land Cr$}
        \For {$i \in [n]$}
            \If {$v_{a_i} > v_m + w_i$} \Comment{Fails Cr}
                \State \Return  \outpt{None}
            \EndIf
        \EndFor
        \State \Return  \outpt{True}
    \EndProcedure
\end{restatable}



\subsubsection{\Pmksp[SM \land Eq \land Cr]}

Since the algorithm for \Pmksp[SM \land Cr] produces a single feasible candidate solution (if it exists), we simply need to additionally check for $Eq$ in the final part of the algorithm. This can be done in $O(n)$ by checking whether all players with the same weight incur the same cost.





\subsection{$\boldsymbol{WOE \land SM}$}

For completeness, we briefly discuss properties containing $WOE \land SM$. For $WOE \land SM$, we can base our subroutines on those for $SM$, with the added requirement of no weak ordered envy between players in consecutive resources. 
A similar argument applies for $WOE \land SM \land Eq$. For $WOE \land SM \land Cr$ and $WOE \land SM \land Eq \land Cr$, we can use the corresponding \textsc{First\_k} and \textsc{Subsequent\_k} of either the $WOE$ version or the $SM$ version, and simply check for satisfaction of the other condition in the subroutine \textsc{Satisfies}. 
This is sufficient due to the uniqueness of a satisfying assignment, if one exists. 
In all cases, a solution is found in time $O(n \log n)$. For brevity, we omit these results from \Cref{tab:summary}.

\subsection{Weak Monotonicity}

\Pmksp{} when $P$ contains $WM$ is generally not tractable, except for the special cases of \Psat[WM] and \Psat[WM \land Eq], where the simple assignment of placing all players on the same resource trivially satisfies $WM$ and $Eq$. Discussion on the intractability of \Pmksp[WM] is contained in the next section.


\section{Intractable Properties} \label{sec:intractable}

Finally, we discuss properties $P$ for which \Pmksp{} is intractable. 
In contrast to the previous section, these problems are more difficult 
because a satisfying assignment of the lowest makespan may not be contiguous, as the example below shows for the case of $WM$.

\begin{example} 
    Consider the instance with $\ww = (4, 3, 2, 1)$ and $m = 2$. The assignment $\ba$ given by $l(\ba) = \{\{4,1\},\{3,2\}\}$ is of the lowest makespan of~$5$, and satisfies $WM$ since all players incur the same cost. It is, however, not contiguous. The makespan-optimal, contiguous $WM$-assignment is $\{\{4,3\},\{2,1\}\}$, with a makespan of~$7$.
\end{example}

Our reductions are from the known NP-hard decision problem, \textsc{Partition} \cite{Karp1972}. 
%
An instance is given by a vector $\bs = (s_1, ..., s_n)$ of positive integers. It is a Yes-instance if there exists a subset $S \subset [n]$ such that $\sum_{i \in S} s_i =  \sum_{i \in [n] \setminus S} s_i$.

The key of our reduction is to set $S = \sum_i s_i$, and to define the reduced instance 
by $\ww = (\bs, S+1,S+1)\in \mathbb Z_{> 0}^{n+2}$ 
and $m = 2$.
For makespan minimization, we additionally set $t = \frac{3}{2}S + 1$.
Note that hardness for satisfiability of $P$ immediately implies hardness of makespan minimization (by setting $t$ to be sufficiently large) but our second result also covers properties for which we can efficiently solve \Psat.


\begin{restatable}{theorem}{TheoremSatisfaction} \label{thm:reduction_satisfiability}
    Deciding whether a $P$-assignment exists is NP-complete 
    for any property $P$ in
    $\{Eq\land Cr, WM \land Cr, WM \land Eq\land Cr \}$.
\end{restatable}

\begin{restatable}{theorem}{TheoremDecision} \label{thm:reduction_decision_problem}
    Deciding whether a $P$-assignment satisfying a given makespan threshold exists is NP-complete 
    for any property $P$ in
    $\{\top, Eq, Cr, Eq \land Cr, WM, WM \land Eq, WM \land Cr, WM \land Eq \land Cr\}$.
\end{restatable}

Note that an assignment need not be unique if it exists. For example, the two assignments $\ba$ and $\ba'$ where $l(\ba) = \{\{2,1,1\},\{2,1,1\}\}$ and $l(\ba') = \{\{2,2\},\{1,1,1,1\}\}$ are clearly not equivalent, and both satisfy $WM \land Eq \land Cr$.

Finally, let us reflect on why the same reduction for NP-hardness does not work for properties containing $WOE$, $SM$, or $EF$. These properties are not satisfied whenever players of different weights (say $w_i < w_j$) are assigned to different resources $a_i \neq a_j$ of the same load $v_{a_i} = v_{a_j}$, because this implies that $v_{a_i} \not < v_{a_j}$ (which fails $SM$) and $v_{a_i} - w_i > v_{a_j} - w_j$ (which fails $WOE$ and $EF$). A Yes-instance of \textsc{Partition} would, therefore, not result in a Yes-instance of \Pmksp{} under the same reduction.

\section{Conclusion}
\label{sec:conclusion}

In the context of equilibrium refinement, the problem of finding a `fair' Nash equilibrium can be phrased as solving the \Psat{} problem where $P$ contains $Cr$ and some additional fairness properties. This problem turned out to be tractable for a large number of fairness classes. 
Even beyond equilibria, this paper shows that the intractability of scheduling on identical machines is lifted when imposing additional, desirable constraints.
The reason is that these constraints impose important structure to the solutions (i.e., contiguity), which sufficiently reduces the search space.

Many research questions follow naturally from our results. 
What is the price of anarchy for the fair coordination problem? For problems where the solution to \Psat{} is guaranteed to exist, how does its optimal solution compare to that of the unconstrained case? This gives a measure of the efficiency-fairness trade-off. In the context of communication partition, the tractable fairness conditions are useful for efficiently finding the optimal partition, by removing the combinatorial blow-up at the coalitional level. However, there may still exist combinatorial difficulty at the partition level, so further investigation is necessary. 
What is the computational complexity of finding the optimal partition associated with these properties?





\bibliographystyle{AAMAS/ACM-Reference-Format} 
\bibliography{references}


\clearpage

\appendix
\section*{Appendix}

\section{Deferred Proofs} \label{sec:deferred_proofs}



\subsection{Property Relationships}

\PropRelations*

\begin{proof}
    
    ($EF \implies Eq$): In the case where $w_i = w_j$, $EF \implies v_{a_i} = v_{a_j} \implies Eq$.
    
    ($EF \implies WOE$): For any $i,j : w_i < w_j$, $EF \implies v_{a_i} - w_i = v_{a_j} - w_j$, which satisfies $WOE$. 
    
    ($WOE \implies WM$): For any $i,j : w_i < w_j$, if $a_i \neq a_j$, then $WOE \implies v_{a_i} - w_i \leq v_{a_j} - w_j \implies v_{a_i} < v_{a_j}$, which satisfies $WM$; if $a_i = a_j$, then $v_{a_i} = v_{a_j}$, which also satisfies $WM$.
    
    ($EF \implies SM$): For any $i,j : w_i < w_j$, $EF \implies v_{a_i} - w_i = v_{a_j} - w_j \implies v_{a_i} < v_{a_j}$, which satisfies $SM$. 
    ($SM \implies WM$): For any $i,j : w_i < w_j$, $SM \implies v_{a_i} < v_{a_j}$, which satisfies $WM$.
    
    ($WOE \land Eq \implies OE$): For any $i,j \in [n]$, consider two cases. If $w_i = w_j$, then $Eq \implies v_{a_i} = v_{a_j}$ which satisfies $OE$; if $w_i < w_j$, then $WOE \implies v_{a_i} - w_i = v_{a_j} - w_j$, which satisfies $OE$.
    
    ($OE \implies WOE \land Eq$): $OE$ implies that for any $i, j \in [n]: w_i = w_j$, $v_{a_i} - w_i \leq v_{a_j} - w_j$ and $v_{a_j} - w_j \leq v_{a_i} - w_i$, which together implies $v_{a_i} = v_{a_j}$, so $WOE \land Eq$ is satisfied. For any $i, j \in [n]:w_i < w_j$, $WOE \land Eq$ is clearly satisfied.
    
    ($WM \land Eq \implies M$): For any $i,j \in [n]$, consider two cases. If $w_i = w_j$, then $Eq \implies v_{a_i} = v_{a_j}$ which satisfies $M$; if $w_i < w_j$, then $WM \implies v_{a_i} \leq v_{a_j} - w_j$, which also satisfies $M$.
    
    ($M \implies WM \land Eq$): $M$ implies that for any $i, j \in [n]: w_i = w_j$, $v_{a_i} \leq v_{a_j} $ and $v_{a_j} \leq v_{a_i} $, which together implies $v_{a_i} = v_{a_j}$, so $WM \land Eq$ is satisfied. For any $i, j \in [n]:w_i < w_j$, $M \implies v_{a_i} \leq v_{a_j}$, which satisfies $WM \land Eq$.
\end{proof}

\subsection{Tractable Properties}

\LemWOEContiguous*
\begin{proof}
    Suppose on the contrary there exists a $WOE$-satisfying, non-contiguous assignment where, under any non-increasing ordering of players and resources, we can find some $i,j \in [n]$ such that $w_i < w_j$ and $a_i < a_j$ (i.e., there is no canonical ordering). Then $v_{a_i} \geq v_{a_j}$, and therefore $v_{a_i} - w_i > v_{a_j} - w_j$, which contradicts $WOE$.
\end{proof}


\LemWOEGlobal*

\begin{proof}
    Given a contiguous assignment $\ba$, consider three consecutive resources indexed $x$, $x+1$ and $x+2$ in its canonical ordering (i.e., $i < j \implies w_i \geq w_j$, and $v_x \geq v_{x+1} \geq v_{x+2}$), and assume that the $WOE$ condition holds between players assigned to any two adjacent resources. I.e., $v_{x+1} - \underline w_{x+1} \leq v_{x} - \overline w_{x}$ and $v_{x+2} - \underline w_{x+2} \leq v_{x+1} - \overline w_{x+1}$, where $\underline w_{x+1}$ and $\underline w_{x+2}$ are the lightest players assigned to $x+1$ and $x+2$ respectively, and $\overline w_{x}$ and $\overline w_{x+1}$ are the heaviest players assigned to $x$ and $x+1$ respectively. Because players are in canonical ordering, $\overline w_{x+1} \geq \underline w_{x+1}$, and therefore $v_{x+2} - \underline w_{x+2} \leq v_{x+1} - \overline w_{x+1} \leq v_{x+1} - \underline w_{x+1} \leq v_{x} - \overline w_{x}$, which means that no player assigned to $x+2$ weakly envies anyone assigned to $x$. In other words, in a canonical ordering, the $WOE$ condition is transitive, and therefore $WOE$ holds between all players across all resources.
\end{proof}


\TheoremWOE*
\begin{proof}
    We shows that \Cref{alg:generic}, combined with the relevant subroutines, solves \Pmksp{} for each property $P$ stated in the theorem. We do so by showing that the algorithm terminates, is sound (i.e., if it returns an assignment, then this assignment has the required property), is complete (i.e., if an assignment with the required property exists, then such an assignment is returned by the algorithm), and runs in time $O(n \log n)$ or $O(n^2)$ respectively. 

    Termination: In each case of $P$, each iteration of \textsc{Subsequent\_k} assigns players to a resource, or terminates if no feasible $k$ is found. Since there are finitely many players and resources, the algorithm always terminates. 

    For soundness, completeness, and running time, we consider each case of $P$:
    
    $\boldsymbol{P = WOE}$
    
    Soundness: By construction, at each iteration of the assignment, no player envies any player assigned to the previous resource. It is sufficient to check that the lightest player of the current assignment does not envy the heaviest player of the previous assignment. Notice also that the $WOE$ condition implies that the current assignment is such that its load less its lightest player is no greater than the load less its heaviest player in the previous assignment. This implies that the load less than its heaviest player in the current assignment is no greater than the load less its heaviest player in the previous assignment, and therefore the $WOE$-satisfaction condition is transitive across iterations, and that it is sufficient to check for $WOE$-satisfaction between consecutive resources. Also notice that the sequentially assigned load on resources are non-increasing due to $WOE$, so it is sufficient to check that the load of the first resource is no greater than the makespan threshold. 
    
    Completeness: Suppose for contradiction that the algorithm does not return an assignment when a solution exists, which we call $\ba$. We know that $\ba$ must be contiguous due to \Cref{lem:WOE_contiguous}. 
    Since the algorithm does not return an assignment, it had run out of available resources to assign players to. This implies $\ba$ must have assigned more players to some resources than the algorithm. However, no additional players could be assigned to the first resource due to the makespan constraint, and no more players could be assigned to each subsequent resource due to the $WOE$ constraint. 
    
    Running time: Each iteration of assigning players to the next resource takes time $O(k)$ to find the maximum number of players that can be assigned at the makespan threshold or to satisfy the $WOE$ condition, where $k$ is the number of players assigned. The sum of $k$ across all iterations is $n$, so the assignment part of the algorithm is in $O(n)$. This leaves the sorting of players as the dominant running time, at $O(n \log n)$.
    
    $\boldsymbol{P = WOE \land Eq}$
    
    Soundness: By construction, if the algorithm outputs an assignment, then the assignment is of makespan at most $t$ (the first and most-loaded resource satisfies this constraint), satisfies $WOE$ through the transitivity of $WOE$-satisfaction across iterations, and satisfies $Eq$ due to the fact that all players of the same weight are eventually distributed across resources. 
    
    Completeness: Suppose for contradiction that the algorithm does not return an assignment when a solution exists, which we call $\ba$. Because the algorithm only fails to return an assignment if it runs out of resources while there are still players to be assigned, it must be the case that $\ba$ assigns more players to some resource than the algorithm. This implies that one of the conditions for selecting $k$ is breached, and thus $\ba$ cannot satisfy $WOE \land Eq$.


    Running time: Because each iteration of assigning players to the next resource takes up to $O(n)$ operations (may require checking all $i \in [n]$ to find the maximum feasible $i$), the process of finding a candidate takes time $O(n^2)$. This dominates the running time of sorting players ($O(n \log n)$).

    $\boldsymbol{P = WOE \land Cr}$
    
    Soundness: By construction, if the algorithm outputs an assignment, then its makespan is at most $t$ (this is an explicit constraint on the first assigned resource, and holds implicitly through the $WOE$ bound for subsequent assignments), and has been checked for $WOE \land Cr$-satisfaction by the subroutine $Satisfies(a, P = WOE \land Cr)$. 
    
    Completeness Suppose for contradiction that the algorithm does not return an assignment when a solution exists, which we call $\ba$. Then $\ba$, in its canonical ordering, must assigned to some resource that does not correspond to the $k$ players prescribed by the algorithm. Consider the first resource (in canonical ordering) that fails this. The assignment of players to this resource fails to satisfying some necessary condition of $WOE$ or $Cr$, or have makespan greater than $t$, and therefore $\ba$ is not a solution to the problem.

    Running time: Because each iteration of assigning players to the next resource takes $O(k)$ operations (by calculating, for each value up to $k$, $W[1:k-1]$, $\frac{W - W[1:k]}{m-1}$ and $W[1:k]$), and that the sum of $k$s across iterations is $n$, the assignment part of the algorithm runs in time $O(n)$. This leaves the sorting of players as the dominant running time, at $O(n \log n)$.
    
    $\boldsymbol{P = WOE \land Eq \land Cr}$
    
    Soundness and completeness: This is essentially the same as the argument for $WOE \land Cr$, but with the added check for $Eq$.

    Running time: Compared to $WOE \land Cr$, the additional check for $Eq$ is in time $O(n)$, which again leaves the sorting of players as the dominant running time, at $O(n \log n)$.
\end{proof}

\LemWOEEq*
\begin{proof}
    Suppose for contradiction there exists some player $k$ s.t. $a_k = a_i$ and $w_k \neq w_i$, and let $v = v_{a_i} = v_{a_j}$ be the load on resources $a_i$ and $a_j$, which must be the same due to $Eq$. If $w_k > w_i$, then player $j$ would weakly envy $k$, since $v - w_k < v - w_j$, so $WOE$ is not satisfied. If $w_k < w_i$, then player $k$ would weakly envy $j$, since $v - w_j < v - w_k$, so $WOE$ is again not satisfied. 
\end{proof}


\LemUniquek*
\begin{proof}
    Suppose on the contrary that $k$ and $k + l$ both satisfy the bounds, where $l$ is a positive integer. Then by combining the two sets of bounds we have,
    \begin{align*} 
        &W[1:k-1] \leq \frac{W - W[1:k]}{m-1} \leq W[1:k] \\ 
        &\leq W[1:k+l-1] \leq \frac{W - W[1:k+l]}{m-1} \leq W[1:k+l]
    \end{align*}
    which cannot be true since $W[1:k] < W[1:k+1]$, which implies $\frac{W - W[1:k]}{m-1} > \frac{W - W[1:k+1]}{m-1}$. 
\end{proof}


\LemSMContiguous*

\begin{proof}
    Suppose on the contrary there exists an $SM$-assignment where, under any non-increasing ordering of players and resources, we can find some $i,j$ such that $w_i < w_j$ and $a_i < a_j$. Then $v_{a_i} \geq v_{a_j}$, which contradicts $SM$.
\end{proof}


\TheoremSM*

\begin{proof}
    We show that \Cref{alg:generic}, combined with the relevant subroutines, solves \Pmksp{} for each property $P$ stated in the theorem. We do so by showing that the algorithm terminates, is sound (i.e., if it returns an assignment, then this assignment has the required property), is complete (i.e., if an assignment with the required property exists, then such an assignment is returned by the algorithm), and runs in time $O(n \log n)$.

    Termination: In each case of $P$, each iteration of \textsc{Subsequent\_k} assigns some players to a resource, or terminates if no feasible $k$ is found. Since there are finitely many players and resources, the algorithm always terminates. 

    For soundness, completeness, and running time, we consider each case of $P$:
    
    $\boldsymbol{P = SM}$
    
    Soundness: By construction, at each iteration of the assignment, $SM$ is maintained by ensuring that consecutive resources have non-increasing loads if players across resources have the same weights, or decreasing loads if players across resources have different weights. And since players are ordered in non-increasing order of weight, $SM$ is transitive across resources. The makespan threshold is maintain by ensuring that the first (and thus most-loaded) resource has load at most $t$.
    
    Completeness: Suppose for contradiction that the algorithm does not return an assignment when a solution exists, which we call $\ba$. We know that $\ba$ must be contiguous due to \Cref{lem:SM_contiguous}. Since the algorithm does not return an assignment, it had run out of available resources to assign players to. This implies $\ba$ must have assigned more players to some resources than the algorithm. However, no additional player could be assigned to the first resource due to the makespan constraint, and no additional player could be assigned to each subsequent resource due to the $SM$ constraint. 
    
    Running time: Because each iteration of assigning players to the next resource takes time $O(k)$ to find the maximum feasible candidate, where $k$ sums to $n$ across iterations, the assignment part of the algorithm is in $O(n)$. This leaves the sorting of players as the dominant running time, at $O(n \log n)$.
    
    $\boldsymbol{P = SM \land Eq}$
    
    Soundness: Similar to case $P = SM$, by construction, if the algorithm outputs an assignment, then the assignment is of makespan at most $t$ and satisfies $SM$. Moreover, the additional assignment rule ensures that all players of the same weight incur the same cost, and thus satisfying $Eq$. 
    
    Completeness: Suppose for contradiction that the algorithm does not return an assignment when a solution exists, which we call $\ba$. Because the algorithm only fails to return an assignment if it runs out of resources while there are still players to be assigned, it must be the case that $\ba$ assigns more players to some resource than the algorithm. This implies that one of the conditions for selecting $k$ is breached, and thus $\ba$ cannot satisfy $SM \land Eq$.

    Running time: For each iteration of players assignment, the candidate $k$ can be computed by finding the largest divisor of $\#(w_i)$ that is below the threshold, which takes at most $O(\sqrt \#(w_i))$ steps. 
    Since $n = \sum_{i \in \supp(\ww)} \#(w_i)$, the total number of operations across all iterations is at most $O(n)$ (if $\#(w_i) = 1$ for all $w_i \in \supp(\ww)$).
    This leaves the sorting of players as having the dominant running time, at $O(n \log n)$.

    $\boldsymbol{P = SM \land Cr}$
   
    Soundness: By construction, if the algorithm outputs an assignment, then its makespan is at most $t$ (this is an explicit constraint on the first assigned resource, and holds implicitly through the $SM$ bound for subsequent assignments), and has been checked for $SM \land Cr$-satisfaction by the subroutine $Satisfies(a, P = SM \land Cr)$. 
    
    Completeness: Suppose for contradiction that the algorithm does not return an assignment when a solution exists, which we call $\ba$. Then $\ba$, in its canonical ordering, must assigned to some resource that does not correspond to the $k$ players prescribed by the algorithm. Consider the first resource (in canonical ordering) that fails this. The assignment of players to this resource fails to satisfying some necessary condition of $SM$ or $Cr$, or have makespan greater than $t$, and therefore $\ba$ is not a solution to the problem.

    Running time: Each iteration of assigning players to the next resource takes $O(1)$ operations, by checking if the candidate $k = \lceil W / m w_1 \rceil$ satisfies the required bounds, and thus the process of generating a candidate assignment takes time $O(n)$. This leaves the sorting of players as the dominant running time, at $O(n \log n)$.
    
    $\boldsymbol{P = SM \land Eq \land Cr}$
    
    Soundness and completeness: This is essentially the same as the argument for $SM$, but with the added check for $Eq$.

    Running time: Compared to $SM \land Cr$, the additional check for $Eq$ is in time $O(n)$, which again leaves the sorting of players as the dominant running time, at $O(n \log n)$.

\end{proof}

\subsection{Intractable Properties}

In this section, we provide our hardness results.
Note that for any property and conjunction of properties, membership in NP is straightforward.
A $P$-assignment can be verified in polynomial time by comparing weights and loads on assigned resources for any pair of players.
Hence, we omit this from the proofs and focus on hardness.



\TheoremSatisfaction*

\begin{proof}
    Consider the following reduction from \textsc{Partition} to the \Psat{} problem. Let $S = \sum_i s_i$, and let the instance of the \Psat{} problem be $w = s \cup \{S+1, S+1\}$, and $m = 2$. We show that if a partition exists (i.e., $s$ can be partitioned into $s'$ and $s \setminus s'$ of equal sums), then the assignment $((s' \cup \{S+1\}), (s \setminus s' \cup \{S+1\}))$ satisfies $P$, and conversely, if the assignment $((s' \cup (S+1)), (s \setminus s' \cup (S+1)))$ satisfies $P$, then a partition exists.

    $\boldsymbol{P = Eq \land Cr}$
    
    ($\Rightarrow$) Clearly, if a partition exists, then the assignment places the same load $\sum_{i \in s'} s_i + (S+1)$ on both resources. This implies that all players incur the same cost, so that no player has a credible deviation, and all players of the same weight incur the same cost, so the assignment is $Eq \land Cr$-satisfying. ($\Leftarrow$) Consider two cases: if the two players of weight $S+1$ share the same resource, then the total load on the other resource cannot exceed $S$, so that both players of weight $S+1$ would prefer to be assigned to another resource, which fails $Cr$, so this cannot be the case. If the two players of weight $S+1$ are assigned to different resources, then they must incur the same cost, meaning that the other players are assigned evenly between the two resources. And thus a partition exists.

    $\boldsymbol{P = WM \land Cr}$ 
    
    ($\Rightarrow$) Clearly, if a partition exists, then the assignment places the same load $\sum_{i \in s'} s_i + (S+1)$ on both resources. This implies that all players incur the same cost, which is $WM \land Cr$-satisfying. ($\Leftarrow$) Consider two cases: If the two players of weight $S+1$ share the same resource, then by the same argument as above, $Cr$ fails, so this cannot be the case. If the two players of weight $S+1$ are assigned different resources, then the load on the two resources must be equal, otherwise $WM$ would not hold between any player $i \in S$ from the heavy-load resource and the $S+1$ player on the light-load resource. This means that players in $S$ are divided evenly between the two resources. And thus a partition exists.

    $\boldsymbol{P = WM \land Eq\land Cr}$ 
    
    ($\Rightarrow$) Same argument as above. ($\Leftarrow$) Consider two cases: If the two players of weight $S+1$ share the same resource, then by the same argument as above, $Cr$ fails, so this cannot be the case. If the two players of weight $S+1$ are assigned different resources, then the load on the two resources must be equal in order to satisfy $Eq$. This means that players in $S$ are divided evenly between the two resources. And thus a partition exists.

    In each of the above cases, the \Psat{} problem is also in NP: given a candidate solution, we can check the polynomial number of constraints imposed by these conditions, each of which takes polynomial time.
\end{proof}



\TheoremDecision*

\begin{proof}
    Consider the following reduction from \textsc{Partition} to a \Pmksp{} problem. Let $S = \sum_i s_i$, and let the instance to the \Pmksp{} problem be $w = s \cup (S+1, S+1)$, $m = 2$, with parameter $t = \frac{3}{2}S + 1$. We want to show that if a partition exists, then the assignment $((s' \cup \{S+1\}), (s \setminus s' \cup \{S+1\}))$ has makespan at most $t$ and satisfies $P$, and conversely, if the assignment $((s' \cup (S+1)), (s \setminus s' \cup (S+1)))$ has makespan $\leq t$ and satisfies $P$, then a partition exists.

    ($\Rightarrow$) If a partition exists, then the assignment places the same load $S/2 + (S+1)$ on both resources, so the makespan is at the decision threshold $t$. Moreover, all player incur the same cost, so it is easy to verify that $P \in \{\top, Eq, Cr, Eq \land Cr, WM, WM \land Eq, WM \land Cr, WM \land Eq \land Cr\}$ is satisfied. ($\Leftarrow$) If an assignment meets the threshold, then it must be the case that the two players of weight $S+1$ are assigned to different resources (otherwise the threshold would be breached), and that the players in $s$ are split evenly across the two resources, which means that a partition exists.

    Since each of the conditions above can be checked in polynomial time, each of the corresponding \Psat{} problems are also in P. We can thus summarize the above results as follows.
\end{proof}

\clearpage
\section{Sequential Contiguous Assignment and its subroutines} \label{sec:algorithm}

For completeness, we provide the SCA algorithm here together with all its subroutines.

\addtocounter{algorithm}{-1}
\begin{algorithm}
    \caption{Sequential Contiguous Assignment (SCA)}
    \begin{algorithmic}[1]
    \Statex \textbf{Input:} A weight vector $\ww = (w_1, w_2, \dots, w_n)$, a number of resources $m$, a property $P$, and a makespan threshold~$t$.
    \Statex \textbf{Output:} A contiguous $P$-assignment of makespan at most $t$, if it exists; \outpt{None} otherwise.

    \vspace{0.5em}
    \State Sort and index $\ww$ in non-increasing order of weights.

    \vspace{0.5em}
    \State $k \gets \textsc{First\_k}(P, t, \ww, m)$
    \State $a_j \gets 1$ for all $j \in [1:k]$ \Comment{Assign players to 1st resource}
    \State $x \gets 2$ \Comment{Define resource counter}
    \State $i \gets k+1$ \Comment{Define player counter}
    \State $t \gets W[1:k]$ \Comment{Update load threshold}
    
    \vspace{0.5em}
    \While{$x \leq m$ and $i \leq n$}
        \State Compute relevant input $z$
        \State $k \gets \textsc{Subsequent\_k}(P, t, 
        \ww[i:n], m-x+1, z)$
        \State $a_j \gets x$ for all $j \in [i,i+k]$ \Comment{Assign players to resource}
        \State $x \gets x+1$ \Comment{Update resource counter}
        \State $i \gets i+k$ \Comment{Update player counter}
        \State $t \gets W[i: i+k-1]$ \Comment{Update load threshold}
    \EndWhile

    \vspace{0.5em}
    \If{$i < n$} \Comment{Some players remain unassigned}
        \State \Return \outpt{None} 

    \vspace{0.5em}
    \Else \Comment{All players assigned to resources}
        \State $a \gets (a_1, ..., a_n)$
        \If{\textsc{Satisfies}$(\ba, P)$} 
        \Comment{Test $\ba$ for $P$}
            \State \Return $a = (a_1, ..., a_n)$
        \Else
            \State \Return \outpt{None}
        \EndIf
    \EndIf
    \end{algorithmic}
\end{algorithm}

\subsection{$\boldsymbol{P = WOE}$}~

\WOEFirstk*
\WOESubsequentk*
\begin{restatable}{algorithmic}{WOESatisfies}
    \Procedure{Satisfies}{$\ba, P = WOE$}
        \State \Return \outpt{True}
    \EndProcedure
\end{restatable}

\subsection{$\boldsymbol{P = WOE \land Eq}$}~

\WOEEqFirstk*
\WOEEqSubsequentk*

\begin{restatable}{algorithmic}{WOEEqSatisfies}
    \Procedure{Satisfies}{$\ba, P = WOE \land Eq$}
        \State \Return  \outpt{True}
    \EndProcedure
\end{restatable}

\subsection{$\boldsymbol{P = WOE \land Cr}$}~

\WOECrFirstk*
\WOECrSubsequentk*

\begin{restatable}{algorithmic}{WOECrSatisfies}
    \Procedure{Satisfies}{$\ba, P = WOE \land Cr$}
        \For {$i \in [n]$}
            \If {$v_{a_i} > v_m + w_i$} \Comment{Fails Cr}
                \State \Return  \outpt{None}
            \EndIf
        \EndFor

        \For {$x \in [m-1]$}
            \State Set $\overline w = \max \{ w_i : a_i = x\}$
            \State Set $\underline w = \min \{ w_i : a_i = x+1\}$
            \If {$v_{x+1} - \underline w > v_x - \overline w$} \Comment{Fails WOE}
                \State \Return  \outpt{None}
            \EndIf
        \EndFor

        \State \Return \outpt{True}
    \EndProcedure
\end{restatable}

\subsection{$\boldsymbol{P = WOE \land Eq \land Cr}$}~

\begin{restatable}{algorithmic}{WOEEqCrFirstk}
    \Procedure{First\_k}{$P = WOE \land Eq \land Cr, t, \ww, m$}
        \If{ $\exists k : W[1:k-1] \leq \frac{W - W[1:k]}{m-1} \leq W[1:k] \leq t\}$}
            \State \Return  $k$    
        \Else{}
            \State \Return  \outpt{None}
        \EndIf
    \EndProcedure
\end{restatable}

\begin{restatable}{algorithmic}{WOEEqCrSubsequentk}
    \Procedure{Subsequent\_k}{$P = WOE \land Eq \land Cr, t, \ww, m, z$}
    \If { $m = 1$} \Comment{Final resource} 
        \State \Return  $k = |\ww|$
    \EndIf
    \If { $\exists k : W[1:k-1] \leq \frac{W - W[1:k]}{m-1} \leq W[1:k] \leq t \}$}
        \State \Return  $k$
    \Else{}
        \State \Return  \outpt{None}
    \EndIf
    \EndProcedure
\end{restatable}

\begin{restatable}{algorithmic}{WOEEqCrSatisfies}
    \Procedure{Satisfies}{$\ba, P = WOE \land Eq \land Cr$}
        \For {$i \in [n]$}
            \If {$v_{a_i} > v_m + w_i$} \Comment{Fails $Cr$}
                \State \Return  \outpt{None}
            \EndIf
        \EndFor

        \For {$i \in [n-1]$}
            \If {$w_i = w_{i+1} \land v_{a_i} \neq v_{a_{i+1}}$} \Comment{Fails $Eq$}
                \State \Return  \outpt{None}
            \EndIf
        \EndFor

        \For {$x \in [m-1]$}
            \State Set $\overline w = \max \{ w_i : a_i = x\}$
            \State Set $\underline w = \min \{ w_i : a_i = x+1\}$
            \If {$v_{x+1} - \underline w > v_x - \overline w$} \Comment{Fails WOE}
                \State \Return  \outpt{None}
            \EndIf
        \EndFor
        
        \State \Return  \outpt{True}
    \EndProcedure
\end{restatable}


\subsection{$\boldsymbol{P = SM}$}~

\SMFirstk*
\SMSubsequentk*

\begin{restatable}{algorithmic}{SMSatisfies}
    \Procedure{Satisfies}{$\ba, P = SM$}
        \State \Return \outpt{True}
    \EndProcedure
\end{restatable}

\subsection{$\boldsymbol{P = SM \land Eq}$}~

\SMEqFirstk*
\SMEqSubsequentk*

\begin{restatable}{algorithmic}{SMEqSatisfies}
    \Procedure{Satisfies}{$\ba, P = SM \land Eq$}
        \State \Return  \outpt{True}
    \EndProcedure
\end{restatable}

\subsection{$\boldsymbol{P = SM \land Cr}$}~

\SMCrFirstk*
\SMCrSubsequentk*
\SMCrSatisfies*

\subsection{$\boldsymbol{P = SM \land Eq \land Cr}$}~

\begin{restatable}{algorithmic}{SMEqCrFirstk}
    \Procedure{First\_k}{$P = SM \land Eq \land Cr, t, \ww, m$}
        \If{ $\exists k : (k-1) w_1 \leq \frac{W - k w_1}{m-1} \leq k w_1 \leq t\}$}
            \If {$ k \leq \#(w_1)$}
                \State \Return $k$      
            \Else{}
                \State \Return  \outpt{None}
            \EndIf
        \Else{}
            \State \Return  \outpt{None}
        \EndIf
    \EndProcedure
\end{restatable}

\begin{restatable}{algorithmic}{SMEqCrSubsequentk}
    \Procedure{Subsequent\_k}{$P = SM \land Eq \land Cr, t, \ww, m, z$}
    \State Let $z$ be the weight of players in previous resource
    \If { $z = w_1$} \Comment{Load must be $\leq t$}
        \If { $m = 1$} \Comment{Final resource}
            \If { $w_1 = w_{|\ww|}$}
                \State \Return  $k = |\ww|$
            \Else{}
                \State \Return  \outpt{None}
            \EndIf
        \EndIf
        \If { $\exists k : (k-1) w_i \leq \frac{W - k w_i}{m-1} \leq k w_i \leq t \land k \leq \#(w_1)$}
            \State \Return $k$      
        \Else{}
            \State \Return  \outpt{None}
        \EndIf
    \Else {} \Comment{Load must be $< t$}
        \If { $m = 1$} \Comment{Final resource}
            \If { $w_1 = w_{|\ww|}$ and $W[1:|\ww|] < t$}
                \State \Return $|\ww|$
            \Else{}
                \State \Return  \outpt{None}
            \EndIf
        \EndIf
        \If { $\exists k : (k-1) w_i \leq \frac{W - k w_i}{m-1} \leq k w_i < t \land k \leq \#(w_1)$}
            \State \Return  $k$      
        \Else{}
            \State \Return  \outpt{None}
        \EndIf
    \EndIf
    \EndProcedure
\end{restatable}

\begin{restatable}{algorithmic}{SMEqCrSatisfies}
    \Procedure{Satisfies}{$\ba, P = SM \land Eq \land Cr$}
        \For {$i \in [n]$}
            \If {$v_{a_i} > v_m + w_i$} \Comment{Fails Cr}
                \State \Return  \outpt{None}
            \EndIf
        \EndFor

        \For {$i \in [n-1]$}
            \If {$w_i = w_{i+1} \land v_{a_i} \neq v_{a_{i+1}}$} \Comment{Fails $Eq$}
                \State \Return  \outpt{None}
            \EndIf
        \EndFor
        
        \State \Return  \outpt{True}
    \EndProcedure
\end{restatable}


\clearpage
\section{Envy-Freeness} \label{sec:EF}

\subsection{\Pmksp[\boldsymbol{EF}]}

\cite{Lee2025PartitionEquilibria} showed that $EF$ is a strong condition that can only be satisfied if and only if there exists some value $b \geq 0$ where for all $w_i \in \supp(\ww)$,

\begin{align*}
    w_i & \mid b \tag{common multiple} \\
    \frac{b}{w_i} + 1 & \mid \#(w_i) \tag{balanced loads} \\
    \sum_{w_i \in \supp(\ww)} \frac{\#(w_i)}{b/w_i + 1} &\leq m \tag{resource constraint}
\end{align*}


Intuitively, an $EF$-assignment, for some feasible value $b$, is an assignment where each resource is assigned $b/w_i + 1$ players of the same weight $w_i$, and thus a total of $\frac{\#(w_i)}{b/w_i + 1}$ resources are used by players of weight $w_i$, and each resource has the load $b + w_i$.

\Cref{alg:EF} provides a search for feasible values of $b$ and outputs their corresponding assignment.


\begin{algorithm}
    \caption{\Psat[EF]}
    \label{alg:EF}
    \begin{algorithmic}[1]
    \Statex \textbf{Input:} A multset of $n$ weights $w = (w_1, w_2, \dots, w_n)$, an integer number of resources $m$, and a makespan threshold $t$.
    \Statex \textbf{Output:} All $EF$-assignment $\ba$ of makespan at most $t$, if it exists; \outpt{None} otherwise.

    \vspace{0.5em}
    \State Compute $\supp(\ww)$, and $\#(w_i)$ for each $w_i \in \supp(\ww)$.
    \State Initialize the candidate set $b_{cand} = \emptyset$.
    \State Compute $D$, the set of all divisors of $\#(w_1)$ 
    \For{$d \in D$} 
        \State Set $b_{cand} = b_{cand} \cup \{(d-1)w_1\}$.
    \EndFor

    \vspace{0.5em}
    \For{$b \in b_{cand}$} 
        \If{$\sum_{w_i \in \supp(\ww)} \frac{\#(w_i)}{b/w_i + 1} > m$} 
            \State Set $b_{cand} = b_{cand} \setminus \{b\}$ 
            \State \textbf{Continue}
        \EndIf
        \For{$w_i \in \supp(\ww)$}
            \If{$w_i \not \mid b$ or $(b/w_i + 1) \not \mid \#(w_i)$} 
                \State Set $b_{cand} = b_{cand} \setminus \{b\}$
            \EndIf
        \EndFor
    \EndFor

    \vspace{0.5em}
    \State Initialize the set of satisfying assignments $a_{set} = \emptyset$
    \For{$b \in b_{cand}$}
        \For{$w_i \in \supp(\ww)$}
            \State Assign $b/w_i + 1$ players to each resource, for a total of $\frac{\#(w_i)}{b/w_i +1}$ resources.
        \EndFor
        \State Add assignment associated with $b$ to $a_{set}$
    \EndFor
    \State \Return $a_{set}$
    \end{algorithmic}
\end{algorithm}

The algorithm first generates a candidate set of values for $b$, based on the \emph{balanced loads} condition for some weight $w_1$ (lines 4-5). 
For each candidate, it then checks the resource constraint condition (lines 7-8), and for each distinct weight in $\supp(\ww)$, the balanced loads and common multiple conditions (lines 10-11), and remove any that fails. This results in a set of feasible values of $b$, and a set of unique assignments, one for each $b$ value, is returned (lines 13-18).

Because the above algorithm returns all $EF$-assignments, we can simply look at the makespan of the optimal $EF$-assignment (i.e., the makespan of the assignment associated with the smallest feasible $b$, which is $b + w_1$ where $w_1$ is the largest weight) to answer the \Pmksp[EF] problem for a given value of $t$.

Since $\#(w_1) \leq n$, finding all divisors of $\#(w_1)$ is in $O(\sqrt n)$. Because also $\supp(\ww) \leq n$, the loops through all divisors of $w_1$ and $\supp(\ww)$ (lines 6-12, 14-17) takes $O(n \sqrt n)$, which is the running time of this algorithm.

\subsection{\Pmksp[\boldsymbol{EF\land Cr}]}

The \Pmksp[EF\land Cr] problem was solved in \cite{Lee2025PartitionEquilibria}, and requires a simple modification to \Cref{alg:EF}. 
First notice that by construction, an $EF$-Satisfying assignment associated with some value $b$ satisfies the condition $v_{a_i} - w_i = b < v_{a_j}$ for all $i, j \in [n]$. This means no player can benefit from switching to a resource that is already used. So it remains the check that there are no unused resources for a player to switch to for $Cr$ to be satisfied. 

We do this by changing the \emph{resource constraint} condition to an equality, i.e. $\sum_i \frac{\#(w_i)}{b/w_i + 1} = m$, which corresponds to changing the pre-condition in line 7 of the algorithm to $\sum_i \frac{\#(w_i)}{b/w_i + 1} \neq m$. Because at most one value of $b$ can satisfy this condition, the $EF \land Cr$-Satisfying problem is unique if it exists.

\begin{example}
    Consider the instance $w = (3, 3, 3, 2, 2, 2, 2), m = 3$, where $\supp(\ww) = (3, 2), \#(3) = 3, \#(2) = 4$. The only feasible value for $b$ is $6$, with the corresponding $EF$-assignment where $l(\ba) =\{\{3,3,3\},\{2,2,2,2\},\emptyset\}$, with a makespan of $9$. Therefore the answer to \Pmksp[EF] is $\ba$ if $t \geq 9$, and \outpt{None} otherwise. 
    Because there is an unused resource in $\ba$, $\ba$ is not $EF \land Cr$-satisfying. 
\end{example}

\end{document}